\documentclass[journal]{IEEEtran}
\usepackage{cite}
\usepackage{amsmath,amssymb,amsfonts}
\usepackage{amsthm}
\usepackage{algorithm}
\usepackage{algorithmic}
\usepackage{graphicx}
\usepackage{booktabs}
\usepackage{array}
\usepackage{xcolor}
\usepackage{tikz}
\usetikzlibrary{arrows.meta,positioning,calc}
\usepackage[hidelinks]{hyperref}

\newtheorem{theorem}{Theorem}

\newtheorem{proposition}{Proposition}

\theoremstyle{definition}
\newtheorem{definition}{Definition}

\newcommand{\E}{\mathcal{E}}
\newcommand{\M}{\mathcal{M}}
\newcommand{\Pset}{\mathcal{P}}
\newcommand{\F}{\mathcal{F}}
\newcommand{\Th}{\Theta}
\newcommand{\R}{\mathbb{R}}
\newcommand{\tr}{^{\top}}

\begin{document}

\title{Minimum Bounding Sets for N-1-1 Branch Contingencies Under Thermal and Voltage Limits in Linear Network Models}

\author{Timothy~N.~Thacker%
\thanks{T.~N.~Thacker is with the Min H. Kao Department of Electrical Engineering and Computer Science, The University of Tennessee, Knoxville, TN 37996 USA, and with CURENT (Center for Ultra-Wide-Area Resilient Electric Energy Transmission Networks) (e-mail: tthacke3@vols.utk.edu).}}

\markboth{}{Thacker: Exact Minimum Bounding Sets for N-1-1 Branch Contingencies}

\date{October 2026}

\maketitle

\begin{abstract}
N-1-1 contingency analysis studies the loss of one element followed by system adjustments and the loss of a second. The number of contingency pairs increases quadratically with the number of applicable elements, and traditional ranking-based screening methods cannot guarantee that omitted pairs are bounded. This paper determines the minimum bounding set of two-branch-outage pairs: the smallest set that can reproduce the worst value of every violated thermal and voltage limit of a linear network model. In the DC model, effects of an outage pair can be quantified by the extent of the coupling between the outages. The paper measures this by the cosine between the two outages in the network's cycle space, with a bound that splits into one term per outage and is within a near-ideal factor of $\sqrt{1+\varepsilon}$. This argument extends to every linear model in which an outage subtracts a positive semi-definite low-rank term from a symmetric positive definite matrix, with a tight bound and a collapse test. A model with voltage magnitudes is extracted by performing one Newton step on the energy function of the lossless network in rectangular coordinates. A best-first search built on these bounds reproduced brute force on public systems of up to 2{,}869 buses with 130 to 3{,}800 times fewer pair evaluations for flows and 170 to 23{,}700 times fewer for voltages. Gaps in translation to nonlinear models are explored on the IEEE 118-bus system.
\end{abstract}


\section{Introduction}

\IEEEPARstart{N}{-1-1} contingency analysis, required in North America by NERC TPL-001-5.1 and formalized as category P6, asks the planner to identify the events ''expected to produce more severe system impacts,'' with a documented rationale (Requirement R3.4) \cite{nerc_tpl001}. P6 contingencies consider elements from transmission circuits, transformers, shunt devices, and DC poles. This paper only addresses sequences of two-branch outages, circuits and transformers, a subset that already shows the full combinatorial difficulty of the N-1-1 problem.

The evaluation of a subset can be traced back to the contingency selection algorithms of the 1970s and 1980s \cite{mikolinnas1981}. Such traditional approaches aim to rank outage pairs using single-outage severity, as with $|\E|$ eligible elements there are $|\E|(|\E|-1)$ ordered sequences. While single-outage severity can reduce the scope, it cannot bound what it skips and Section~\ref{sec:unsound} shows that a pair can overload a branch that neither single outage affects at all.

This paper asks whether the selection can be made exact: can one compute, defensibly and consistently, the smallest set of branch-outage pairs whose study bounds every other pair, without enumerating the full set? The results of this work indicate that for voltage and thermal limits in linear network models, the answer is yes. System adjustments between outages are included for thermal limit evaluation. The DC model is used for thermal criteria; voltage criteria are evaluated using a symmetric linear AC model where each outage subtracts a positive semi-definite low-rank term from the base system matrix. In both cases the sets are exact for the model and provide a partial answer for the AC equations, with the gaps between DC and AC formulation explored in Section~\ref{sec:ac}.

\subsection{Related Work}

The N-2 screening of Turitsyn and Kaplunovich asserts the pair response as the two single-outage responses plus an interference term, filters candidate pairs with bounds constructed from Line Outage Distribution Factors (LODFs), and does not miss a violating pair in the DC model \cite{turitsyn2012,kaplunovich2013,kaplunovich2016}. Multiple-outage distribution factors follow from the compensation formula \cite{guler2007}. Kaiser, Strake, and Witthaut validated concerns of single-outage ranking limitations for determining violating pairs by showing that two simultaneous failures can amplify, attenuate, or reverse the individual effects and quantified this behavior using the geometric mean of the mutual LODFs \cite{kaiser2020}. Cycle-space and spanning-forest descriptions of outage factors appear in \cite{ronellenfitsch2017,guo2020}, and Lyons demonstrated that the transfer-current matrix is an orthogonal projector \cite{lyons2003}.

\subsection{Contributions}

\begin{enumerate}
\item A separable bound on the interaction between two outages, with constant $|c|/\sqrt{1-|c|}$ in the cycle-space cosine $c$, tight to within $\sqrt{1+\varepsilon}$, with a count of strongly coupled pairs and an extension to multi-branch elements (Section~\ref{sec:algebra}).
\item A best-first search that returns the exact worst pair of every criterion, is instance-optimal for its bounds, and has an a priori work bound (Section~\ref{sec:search}).
\item A Helly-type bound on the second contingencies that decide each first contingency under adjustment (Section~\ref{sec:redispatch}).
\item A class of linear outage models, containing the DC model, in which each outage subtracts a positive semi-definite low-rank term from a symmetric positive definite matrix. For this class the paper gives an exact pair formula, an exact collapse test, and a tight separable interaction bound (Section~\ref{sec:voltage}).
\item A member of this class with voltage magnitudes, built from the energy function of the lossless network in rectangular coordinates (Section~\ref{sec:emodel}).
\item Case studies on five public systems, including a comparison with the filter of \cite{kaplunovich2016} (Section~\ref{sec:cases}).
\item An AC study of all pairs of the IEEE 118-bus system, which compares the voltage model with the power-mismatch linearization and the decoupled model (Sections~\ref{sec:ac} and~\ref{sec:acresults}).
\end{enumerate}

\section{Problem Formulation}\label{sec:problem}

Consider a transmission network with $n$ buses. Branch $\ell$ has susceptance $b_\ell>0$ and signed incidence row $a_\ell\tr$. The incidence matrix is $A$, $W=\mathrm{diag}(b_\ell)$, and $L=A\tr WA$ is the DC Laplacian. The set of contingency elements is $\E$ and the set of monitored branches is $\M$. Every element is a branch: a transmission circuit or a transformer. Shunt devices and single poles of DC lines are not considered. Each element is a single branch until Section~\ref{sec:multi}.

A transmission planning study uses a finite list of study cases, $\Th=\{\theta_1,\dots,\theta_s\}$. Each thermal criterion $\chi=(m,\sigma)$, with a flow direction $\sigma=\pm1$, gives the normalized value $\varphi_\chi(p,\theta)=\sigma f^{p}_m/R_m$ of pair $p=(k,j)$ in case $\theta$, where $f^p_m$ is the post-pair flow and the rating is $R_m$. A pair violates $\chi$ when $\varphi_\chi>1$. Voltage criteria, defined in Section~\ref{sec:voltage}, have the same form with the bus voltage replacing the flow.

\begin{definition}\label{def:bounding}
A set $S$ of sequences \emph{bounds} a criterion family $\F$ if, for every $\theta\in\Th$ and every $\chi\in\F$ with $\max_p\varphi_\chi(p,\theta)>1$, the maximum of $\varphi_\chi(\cdot,\theta)$ over $S$ equals its maximum over all sequences.
\end{definition}

A set bounds $\F$ exactly when it contains a maximizer of every violated criterion. Thus a \emph{minimum bounding set} $S^*$ consists of one maximizer per violated criterion when maximizers are unique, so that $|S^*|$ is at most the number of violated criteria, and otherwise a minimum hitting set of the sets of tied maximizers. Considering ties, it reproduces the worst value of every violated criterion, which is the scenario that corrective actions must cover. The set bounds the criteria of the model in which $\varphi$ is computed, the DC model for thermal criteria, and the energy model of Section~\ref{sec:emodel} for voltage criteria. It must be recalculated after any corrective action changes the network. Section~\ref{sec:acresults} aims to measure how far these sets are from bounding the AC quantities.

\section{Outage Algebra in the Cycle Space}\label{sec:algebra}

\subsection{Known Pair Algebra}

The DC outage algebra is known and stated in the cycle-space form of \cite{ronellenfitsch2017} with \cite{guler2007,kaiser2020,lyons2003} providing the proofs. Let $\Phi_{mk}=b_ma_m\tr L^{+}a_k$ be the branch-to-branch Power Transfer Distribution Factor (PTDF) and $d_{mk}=\Phi_{mk}/(1-\Phi_{kk})$ the LODF. The transfer-current matrix $\Pi=W^{1/2}AL^{+}A\tr W^{1/2}$ is the orthogonal projector onto the range of $W^{1/2}A$, with trace $n-1$ \cite{lyons2003}, so $Q=I-\Pi$ projects onto the weighted cycle space and $\Phi_{kk}=\Pi_{kk}\in(0,1]$, with equality exactly for bridges.  With $q_k=Qe_k$ define the cycle-space cosine, the cycle exposure of a monitored branch, and the outage stress of an element,
\begin{equation}
c_{kj}=\frac{Q_{kj}}{\sqrt{Q_{kk}Q_{jj}}},\quad \xi_m=\sqrt{b_mQ_{mm}},\quad S_k=\frac{f_k}{\sqrt{b_kQ_{kk}}}.
\end{equation}
The single-outage change on $m$ is
\begin{equation}\label{eq:single}
\Delta^k_m=d_{mk}f_k=-\xi_m\,c_{mk}\,S_k .
\end{equation}
If removing $k$ and $j$ leaves the network connected, the compensation formula \cite{guler2007,kaiser2020} gives, for $m\notin\{k,j\}$ and $\kappa_{kj}=d_{kj}d_{jk}$,
\begin{align}
\Delta^{kj}_m&=\frac{(d_{mk}+d_{mj}d_{jk})f_k+(d_{mj}+d_{mk}d_{kj})f_j}{1-\kappa_{kj}}\label{eq:pair}\\
&=-\sqrt{b_m}\,\langle q_m,\,UC^{-1}s\rangle ,\label{eq:pairQ}
\end{align}
where $U=[q_k/\sqrt{Q_{kk}}\;\;q_j/\sqrt{Q_{jj}}]$ has Gram matrix $C=[1\;c;\,c\;1]$ with $c=c_{kj}$, and $s=(S_k,S_j)\tr$. The second line is the cycle-space form \cite{ronellenfitsch2017}. In this notation $\kappa_{kj}=(Q_{kj}/Q_{jj})(Q_{jk}/Q_{kk})=c_{kj}^2$, so the predictor $\sqrt{\kappa_{kj}}$ of collective effects identified in \cite{kaiser2020} is the absolute cosine between two cycle-space vectors, and $\kappa_{kj}=1$ exactly when $\{k,j\}$ is a cut set. Without adjustment the pair change is symmetric in $k$ and $j$, so only unordered pairs require study.

\subsection{The Interaction Bound}

The novel result of this section bounds the part of the pair change that the two single changes do not quantify.

\begin{proposition}[Interaction bound]\label{prop:residual}
Let $\Delta^\times_{kj,m}=\Delta^{kj}_m-\Delta^k_m-\Delta^j_m$ and $c=c_{kj}$. Then
\begin{equation}\label{eq:residual}
|\Delta^\times_{kj,m}|\le\frac{|c|}{\sqrt{1-|c|}}\,\xi_m\|s\|\le\frac{|c|}{\sqrt{1-|c|}}\,\xi_m\bigl(|S_k|+|S_j|\bigr),
\end{equation}
and $|\Delta^{kj}_m|\le\xi_m\sqrt{\mathcal{W}_{kj}}$ with $\mathcal{W}_{kj}=s\tr C^{-1}s\le\|s\|^2/(1-|c|)$.
\end{proposition}
\begin{proof}
By (\ref{eq:single}) and (\ref{eq:pairQ}), $\Delta^\times_{kj,m}=-\sqrt{b_m}\langle q_m,z\rangle$ with $z=U(C^{-1}-I)s$. Because $\|Uy\|^2=y\tr Cy$, $\|z\|^2=s\tr(C^{-1}-2I+C)s$, and at the eigenvalues $1\pm c$ of $C$ the function $\mu\mapsto(1-\mu)^2/\mu$ equals $c^2/(1\pm c)$. Hence $\|z\|\le|c|\|s\|/\sqrt{1-|c|}$, and the Cauchy--Schwarz inequality with $\sqrt{b_m}\|q_m\|=\xi_m$ gives the first bound. The cap follows in the same way from $\|UC^{-1}s\|^2=s\tr C^{-1}s$ and $\lambda_{\min}(C)=1-|c|$.
\end{proof}

The second form of (\ref{eq:residual}) splits into terms for $k$ and $j$, as the search requires. The constant is nearly tight.

\begin{theorem}[Tightness]\label{thm:tight}
Let $C_2^*(\varepsilon)$ and $C_1^*(\varepsilon)$ be the smallest constants with $|\Delta^\times_{kj,m}|\le C\xi_m\|s\|$ and $|\Delta^\times_{kj,m}|\le C\xi_m(|S_k|+|S_j|)$ for every configuration of unit vectors $q_m/\|q_m\|$, $u_k$, $u_j$ with a positive semi-definite Gram matrix and $|c_{kj}|\le\varepsilon$. Then $C_2^*=\varepsilon/\sqrt{1-\varepsilon}$ and $\varepsilon/\sqrt{1-\varepsilon^2}\le C_1^*\le\varepsilon/\sqrt{1-\varepsilon}$.
\end{theorem}
\begin{proof}
The upper bounds are Proposition~\ref{prop:residual}. For $C_2^*$ take $c=\varepsilon$, $s=(1,-1)\tr/\sqrt2$, an eigenvector of $C$ for $1-c$, and $q_m/\|q_m\|=Us/\sqrt{1-c}$. Then $z=\frac{c}{1-c}Us$ and $\langle q_m/\|q_m\|,z\rangle=c/\sqrt{1-c}$. For $C_1^*$ take $c_{mk}=0$, $c=\varepsilon$, $c_{mj}=\sqrt{1-\varepsilon^2}$, and $S_j=0$.
\end{proof}

While extremal configurations have not been shown to arise from a network with positive susceptances, they are nearly attained in practice (Section~\ref{sec:struct}).

\subsection{How Many Pairs Are Strongly Coupled}

\begin{theorem}[Coupling budget]\label{thm:budget}
For every non-bridge $k$, $\sum_{j\ne k}\kappa_{kj}(1-\Phi_{jj})=\Phi_{kk}$. Then for $\varepsilon,\delta\in(0,1)$,
\begin{equation}\label{eq:count}
\bigl|\{(k,j):\kappa_{kj}>\varepsilon^2,\ 1-\Phi_{jj}\ge\delta\}\bigr|<\frac{n-1}{\varepsilon^2\delta}.
\end{equation}
\end{theorem}
\begin{proof}
$Q^2=Q$ gives $\sum_jQ_{jk}^2=Q_{kk}$. Dividing by $Q_{kk}$ and using $\kappa_{kj}Q_{jj}=Q_{jk}^2/Q_{kk}$ gives the identity. Each counted $j$ contributes more than $\varepsilon^2\delta$, and $\sum_k\Phi_{kk}\le\mathrm{tr}\,\Pi=n-1$.
\end{proof}

An element cannot be strongly coupled to many partners that lie on cycles. Partners that lie on almost no cycle ($1-\Phi_{jj}<\delta$) escape the budget, and on the small test systems they are numerous.

\subsection{Elements That Trip Several Branches}\label{sec:multi}

If element $k$ removes the branch set $K$, with disjoint sets for different elements, let $U_k=Q_{:,K}Q_{KK}^{-1/2}$, $s_k=Q_{KK}^{-1/2}W^{-1/2}_{KK}f_K$, $S_k=\|s_k\|$, and $c_{kj}=\|U_k\tr U_j\|$, the cosine of the smallest principal angle between the cycle-space subspaces of the two elements. Because $Q$ is a projector, $U_k\tr U_k=I$, (\ref{eq:pairQ}) holds with $U=[U_k\;U_j]$ and $C=[I\;G;\,G\tr\;I]$, $G=U_k\tr U_j$, and the eigenvalues of $C$ are $1\pm\sigma_i(G)$ and possibly 1. Proposition~\ref{prop:residual} therefore holds verbatim with $|c|$ replaced by $c_{kj}$. On about 3{,}900 random pairs of two-branch elements in each of the three smaller systems the largest ratio of interaction term to bound was 0.998.

\subsection{Why Ranking Rules Cannot Bound}\label{sec:unsound}

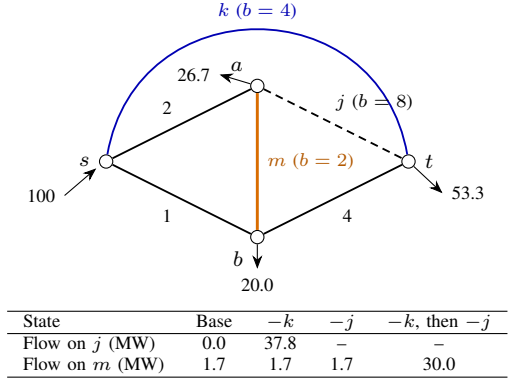
\begin{figure}[!t]
\centering
\begin{tikzpicture}[
  bus/.style={circle,draw,fill=white,inner sep=1.3pt,minimum size=5pt},
  br/.style={line width=0.7pt},
  lab/.style={font=\footnotesize},
  small/.style={font=\scriptsize},
  >={Stealth[length=1.8mm]}]
  \node[bus,label={[lab]left:$s$}] (S) at (0,0) {};
  \node[bus,label={[lab]above left:$a$}] (A) at (2.0,1.0) {};
  \node[bus,label={[lab]below left:$b$}] (B) at (2.0,-1.0) {};
  \node[bus,label={[lab]right:$t$}] (T) at (4.0,0) {};
  \draw[br,color=blue!70!black] (S) to[out=80,in=100,looseness=1.45] node[small,above,pos=0.5]{$k$ ($b=4$)} (T);
  \draw[br] (S) -- node[small,above left,pos=0.5]{2} (A);
  \draw[br] (S) -- node[small,below left,pos=0.5]{1} (B);
  \draw[br,densely dashed] (A) -- node[small,above right,pos=0.45]{$j$ ($b=8$)} (T);
  \draw[br] (B) -- node[small,below right,pos=0.5]{4} (T);
  \draw[br,color=orange!85!black,line width=1.1pt] (A) -- node[small,right,pos=0.5,color=orange!70!black]{$m$ ($b=2$)} (B);
  \draw[->] (S) ++(-0.55,-0.45) node[small,left]{100} -- ++(0.42,0.36);
  \draw[->] (A) -- ++(-0.5,0.15) node[small,left]{26.7};
  \draw[->] (B) -- ++(0,-0.42) node[small,below]{20.0};
  \draw[->] (T) -- ++(0.45,-0.42) node[small,right]{53.3};
  \node[anchor=north,font=\scriptsize] at (2.0,-1.85) {%
    \begin{tabular}{lcccc}
    \hline
    State & Base & $-k$ & $-j$ & $-k$, then $-j$\\
    \hline
    Flow on $j$ (MW) & 0.0 & 37.8 & -- & -- \\
    Flow on $m$ (MW) & 1.7 & 1.7 & 1.7 & 30.0\\
    \hline
    \end{tabular}};
\end{tikzpicture}
\caption{A four-bus example of hidden amplification. Susceptances are in per unit on a 100-MVA base and injections in MW. Neither single outage changes the flow on $m$; the sequence $k$ then $j$ raises it from 1.7 to 30.0~MW (DC model).}
\label{fig:hidden}
\end{figure}

Contingency screening approaches built solely on single-outage impacts can miss violating pairs as two outages can have amplification or attenuation effects beyond the sum of their single effects \cite{kaiser2020}.  Fig.~\ref{fig:hidden} shows the extreme amplification case. Branch $j$ initially carries no flow, so losing it has no impact on the system. Meanwhile, $m$ is balanced with respect to $k$, and when $k$ is removed the redistributed flow is carried by $j$. The loss of $k$ and $j$ forces the flow onto $m$. Here $c_{mk}=0$ and $S_j=0$, and (\ref{eq:pairQ}) gives $\Delta^{kj}_m=\xi_mc\,c_{mj}S_k/(1-c^2)$, which grows without bound as $|c|\to1$. A rule that sees only $f_m$, $f_k$, $f_j$, $\Delta^k_m$, and $\Delta^j_m$ cannot tell $m$ from a branch on a separate loop that the pair leaves unchanged, and a rule that pairs the outages with the largest single impacts on $m$ ranks $k$ and $j$ last. A correct rule needs the mutual coupling $c_{kj}$, which (\ref{eq:pair}) computes from N-1 quantities.

\section{Best-First Search}\label{sec:search}

\subsection{Bounds and Algorithm}

Given a case and criterion $(m,\sigma)$, write $v(p)$ for $\sigma f^p_m$ and split the pairs into a strong set $\Pset_s=\{|c_{kj}|>\varepsilon\}$ and a weak set containing the rest. With $\omega=\varepsilon/\sqrt{1-\varepsilon}$, Proposition~\ref{prop:residual} gives $v(k,j)\le\sigma f_m+w_k+w_j$ for weak pairs, with
\begin{equation}\label{eq:scores}
w_k=\sigma\Delta^k_m+\omega\,\xi_m|S_k| ,
\end{equation}
and by Theorem~\ref{thm:tight} $\omega$ is optimal to within the factor $\sqrt{1+\varepsilon}$. For a strong pair $p$ the cap $b(p)$ is the smaller of $\sigma f_m+\xi_m\sqrt{\mathcal{W}_p}$ and $\sigma f_m+\sigma\Delta^k_m+\sigma\Delta^j_m+|c|\xi_m\|s\|/\sqrt{1-|c|}$. A criterion is immune if
\begin{equation}\label{eq:immunity}
|f_m|+\xi_m\max\Bigl(\sqrt{\tfrac{S_{(1)}^2+S_{(2)}^2}{1-\varepsilon}},\ \max_{p\in\Pset_s}\sqrt{\mathcal{W}_p}\Bigr)\le R_m ,
\end{equation}
where $S_{(1)}$ and $S_{(2)}$ are the two largest outage stresses. By Proposition~\ref{prop:residual} no pair can violate it. Near-radial branches, with small $\xi_m$, are discarded this way.

\begin{algorithm}[!t]
\caption{Best-first search for one criterion}\label{alg:search}
\begin{algorithmic}[1]
\REQUIRE scores $w$, strong pairs with caps $b(p)$, rating $R_m$
\STATE \textbf{if} the criterion is immune by (\ref{eq:immunity}) \textbf{return} no violation
\STATE sort $w$ and the strong caps decreasingly; heap $\gets$ the two largest weak scores; $\hat v\gets-\infty$
\LOOP
  \STATE $t\gets$ larger of $\sigma f_m+$ heap top $w_k+w_j$ and the next strong cap
  \STATE \textbf{if} $t\le\max(\hat v,R_m)$ \textbf{break}
  \STATE pop that pair $p$; for a weak pop, push its two successors
  \STATE skip $p$ if it is strong, islanding, or contains $m$
  \STATE evaluate $v(p)$ by (\ref{eq:pair}); if $v(p)>\hat v$ set $\hat v\gets v(p)$, $p^*\gets p$
\ENDLOOP
\RETURN $(\hat v,p^*)$ if $\hat v>R_m$, else no violation
\end{algorithmic}
\end{algorithm}

The weak stream enumerates pairs in decreasing order of $w_k+w_j$ by the standard heap enumeration of sorted pair sums. The two streams are merged by bound, so pairs leave the search in decreasing order of their bound.

\begin{theorem}[Exactness and instance optimality]\label{thm:correct}
Algorithm~\ref{alg:search} returns $v^*=\max_pv(p)$ and a maximizer whenever $v^*>R_m$, and otherwise reports that no pair violates the criterion. Any algorithm that accesses an instance only through the bounds and exact evaluations, and is correct on every instance consistent with them, must evaluate every pair whose bound exceeds $\max(v^*,R_m)$. Algorithm~\ref{alg:search} evaluates exactly these pairs, plus possibly pairs whose bound equals $\max(v^*,R_m)$.
\end{theorem}
\begin{proof}
Pairs leave in decreasing order of bound, so at termination no unevaluated pair can exceed $\max(\hat v,R_m)$. The lower bound on the work is the argument of the threshold algorithm \cite{fagin2003}: an unevaluated pair whose bound exceeds $\max(v^*,R_m)$ could take any value up to its bound.
\end{proof}

Optimality is dependent on the information in the bounds as the instances in the proof do not need to come from a network. Continuing the search while $t\ge\hat v$ collects the tie set of the criterion, and continuing while $t>R_m$ lists each violating pair. Only the number of pairs evaluated depends on $\varepsilon$, as every $\varepsilon$ gives valid bounds. Thus the worst values returned are not dependent on $\varepsilon$ and the parameter is not a screening threshold. The implementation of Section~\ref{sec:cases} compares bounds and treats pairs with $1-\kappa_{kj}<10^{-9}$ as islanding scenarios.

\subsection{A Priori Work Bound}

\begin{proposition}[Work bound]\label{thm:work}
For criterion $(m,\sigma)$ let $\varphi^0=\sigma f_m/R_m$ and $a_o=\sigma\Delta^o_m/R_m$. If (i) $(w_k+w_j)/R_m\le a_k+a_j+\varrho$ for every weak pair, (ii) $\varphi^0+a_o\le1-\mu$ for every $o$, with $\mu>\varrho$, and (iii) $\sum_oa_o^2\le A^2$, then Algorithm~\ref{alg:search} evaluates at most $N(N-1)/2$ weak pairs for this criterion, with $N=\lfloor A^2/(\mu-\varrho)^2\rfloor$.
\end{proposition}
\begin{proof}
An evaluated weak pair has $\varphi^0+a_k+a_j+\varrho>1$, which with (ii) forces $a_j>\mu-\varrho$ and $a_k>\mu-\varrho$. By (iii) at most $A^2/(\mu-\varrho)^2$ elements qualify.
\end{proof}

With $\hat S=\max_o|S_o|/\sqrt{1-\Phi_{oo}}$ one may take $A=\xi_m\hat S/R_m$ and $\varrho=2\omega\xi_m\hat S/R_m$, because $\sum_oQ_{mo}^2=Q_{mm}$. Then $N\le(2\xi_m\hat S/(\mu R_m))^2$ for $\omega=\mu R_m/(4\xi_m\hat S)$, which depends on how close the branch loads to its N-1 limit, not explicitly on $|\E|$. A tighter count is available before the search: at most $N(N-1)/2$ weak pairs are evaluated with $N=|\{o:w_o>R_m-\sigma f_m-\max_ow_o\}|$.

\section{Voltage Criteria: Linear Models With Low-Rank Outages}\label{sec:voltage}

From nodal analysis, the system matrix can be written as a sum of branch terms. In the DC model, where $L=\sum_ob_oa_oa_o\tr$ the term of branch $o$ is $b_oa_oa_o\tr$. A branch term $b_oa_oa_o\tr$ is low rank, nonzero only in the rows and columns of the branch terminals, and is removed from the sum under outage. In the DC model, the system matrix is symmetric positive definite, every branch term is positive semi-definite, and the injection released by an outage lies in the range of its branch term. This section proves the pair formula and a separable bound for every model with these properties and then constructs a model with voltage magnitudes. For the criterion $V_i\ge V^{\min}$ ($s=-1$) or $V_i\le V^{\max}$ ($s=+1$) at a PQ bus $i$, the value of a pair is $s\Delta V^{kj}_i$ and its threshold is $V^0_i-V^{\min}$ or $V^{\max}-V^0_i$, with $V^0$ the base voltage. A clean base case is assumed: criteria violated before the contingencies are excluded.

\subsection{Low-Rank Outage Models}

\begin{definition}\label{def:lrmodel}
A \emph{low-rank outage model} consists of a symmetric system matrix $H\succ0$ in $\R^{n\times n}$ and, for each element $o$, a branch term $U_o\Sigma_oU_o\tr$ (the sum of the terms of the branches it removes), with $U_o\in\R^{n\times r_o}$ and $\Sigma_o$ diagonal and non-singular, and a vector $z_o\in\R^{r_o}$. After outage set $O$ the state change solves
\begin{equation}\label{eq:lrmodel}
H_{-O}\,\Delta z^O=U_O\Sigma_Oz_O,\qquad H_{-O}=H-U_O\Sigma_OU_O\tr ,
\end{equation}
and monitored quantities are linear functionals $\ell\tr\Delta z$. The branch terms are \emph{semi-definite} if every $\Sigma_o\succ0$.
\end{definition}

The DC model is a low-rank outage model with semi-definite branch terms: $H$ is the grounded Laplacian, $U_o=a_o$, $\Sigma_o=b_o$, and $z_o=a_o\tr\theta$, so the right side of (\ref{eq:lrmodel}) is the injection $f_oa_o$ that cancels the flow of the removed branch.

\begin{theorem}[Pair formula and collapse test]\label{thm:lrpair}
Let $K_O=\Sigma_O^{-1}-U_O\tr H^{-1}U_O$, of order $r_k+r_j$. (i) The inertias satisfy $\mathrm{In}(K_O)=\mathrm{In}(\Sigma_O)+\mathrm{In}(H_{-O})-(n,0,0)$, so $H_{-O}\succ0$ exactly when $K_O$ has the inertia of $\Sigma_O$, and $H_{-O}$ is singular exactly when $K_O$ is. (ii) If $K_O$ is non-singular, $\Delta z^O=H^{-1}U_OK_O^{-1}z_O$.
\end{theorem}
\begin{proof}
(i) is the Haynsworth inertia formula \cite{haynsworth1968} applied to $[H\;U_O;\,U_O\tr\;\Sigma_O^{-1}]$, whose Schur complements are $K_O$ and $H_{-O}$. (ii) is the Woodbury identity.
\end{proof}

Therefore a pair needs a single solve of order $r_k+r_j$, and an inertia count on the same matrix determines if the pair collapses the model. For an element with $H_{-o}\succ0$ write $K_{oo}=V_o\Lambda_oV_o\tr$, $D_o=V_o|\Lambda_o|^{-1/2}V_o\tr$, and $\Omega_o=V_o\,\mathrm{sign}(\Lambda_o)V_o\tr$, which is the identity for a semidefinite branch term. Define the outage stress $s_o=D_oz_o$, the exposure $\gamma_o=D_oU_o\tr H^{-1}\ell$ of a monitored quantity, and the coupling $G_{kj}=-D_kU_k\tr H^{-1}U_jD_j$ with $c_{kj}=\|G_{kj}\|_2$. Then $D_OK_OD_O=C=[\Omega_k\;G_{kj};\,G_{kj}\tr\;\Omega_j]$, the single change is $\ell\tr\Delta z^o=\gamma_o\tr\Omega_os_o$, and the pair change is $\gamma\tr C^{-1}s$ with $\gamma=(\gamma_k,\gamma_j)$ and $s=(s_k,s_j)$. In the DC model $c_{kj}$ is the cycle-space cosine and $s_o$ the outage stress of Section~\ref{sec:algebra}.

\begin{theorem}[Coupling and interaction]\label{thm:lrbound}
Let $H_{-k},H_{-j}\succ0$ and $c=c_{kj}$. (i) If both branch terms are semi-definite and $H_{-kj}\succ0$, then $c<1$. (ii) If $c<1$, then $H_{-kj}$ is non-singular and for every functional $\ell$
\begin{equation}\label{eq:lrbound}
\begin{aligned}
|\ell\tr(\Delta z^{kj}-\Delta z^k-\Delta z^j)|&\le\tfrac{c}{1-c}\|\gamma\|\,\|s\|,\\
|\ell\tr\Delta z^{kj}|&\le\tfrac{1}{1-c}\|\gamma\|\,\|s\|.
\end{aligned}
\end{equation}
(iii) No smaller constant than $c/(1-c)$ holds for all $\gamma$ and $s$.
\end{theorem}
\begin{proof}
(i) By Theorem~\ref{thm:lrpair}, $C\succ0$, and a positive definite matrix with identity diagonal blocks has off-diagonal blocks of norm less than one. (ii) Let $\Omega=\mathrm{diag}(\Omega_k,\Omega_j)$, which is symmetric and orthogonal, and $\tilde G$ the off-diagonal part of $C$, so $C=\Omega(I+\Omega\tilde G)$ with $\|\Omega\tilde G\|=c<1$. Hence $C$ is invertible, $\|C^{-1}\|\le1/(1-c)$, and the interaction term $\gamma\tr(C^{-1}-\Omega)s=-\gamma\tr(I+\Omega\tilde G)^{-1}\Omega\tilde G\Omega s$ is at most $c\|\gamma\|\|s\|/(1-c)$ in magnitude. (iii) With $\Omega=I$, take $\gamma=s$ equal to a unit eigenvector of $C$ for the eigenvalue $1-c$.
\end{proof}

The bound is separable after one more step: $\|\gamma\|\|s\|\le\frac12\sum_{o\in\{k,j\}}(t\|\gamma_o\|^2+\|s_o\|^2/t)$ for every $t>0$. Because the DC system matrix is the sum of its branch terms, their normalized factors are orthonormal in the cycle space. That Gram structure gives the sharper constant of Proposition~\ref{prop:residual} and the budget of Theorem~\ref{thm:budget}. In general $H$ exceeds the sum of its branch terms by an indefinite matrix, and by (iii) the constant $c/(1-c)$ cannot be improved.

\subsection{A Symmetric Linear AC Model}\label{sec:emodel}

Let $x=(e,f)$ be the rectangular bus voltages, $V=e+\mathrm{j}f$, and let $Y_0$ be the bus admittance matrix with every series admittance $y$ replaced by $\mathrm{j}\,\mathrm{Im}\,y$ and all conductances removed. The matrix $B=-\mathrm{j}Y_0$ is Hermitian, which also holds with phase shifters. The network energy $w(x)=-\frac12V^{\mathsf H}BV$ is a real quadratic form, to which branch $o$ contributes $\frac12b_o|V_f/\tau_o-V_t|^2$ through its series element and a non-positive term through its charging. Define
\begin{equation}\label{eq:energy}
\mathcal W(x)=w(x)-\sum_{i\ne r}P_i\,\theta_i(x)-\sum_{i\in\mathrm{PQ}}Q_i\ln|V_i| ,
\end{equation}
with $r$ the slack bus and $P_i+\mathrm{j}Q_i=V_i\overline{(Y_0V)_i}$ at the solved AC base $x_b$, that is, the AC injections with the base conductance losses of each bus held fixed. Differentiation shows that the derivatives of $\mathcal W$ with respect to $\theta_i$ and $\ln|V_i|$ are the active and reactive mismatches of the lossless network, so $x_b$ is a critical point. Voltage-controlled buses carry the constraint $|V_i|=V_i^{\mathrm{set}}$. Let $z$ collect $(\Delta e_i,\Delta f_i)$ at PQ buses and the displacement along the circle $|V_i|=V_i^{\mathrm{set}}$ at voltage-controlled buses, with basis $P$, and let $H=P\tr\nabla^2\mathcal L(x_b)P$ for the Lagrangian $\mathcal L$ of these constraints.

Each $w_o$ is a quadratic form with a constant Hessian $K_o$ on the two terminals: positive semi-definite of rank two for the series element, negative semi-definite for the charging. Removing $O$ replaces $\mathcal W$ by $\mathcal W-\sum_{o\in O}w_o$, and one Newton step on its stationarity conditions from $x_b$ is $H_{-O}\Delta z=P\tr\sum_{o\in O}K_ox_b$, because $\nabla w_o(x)=K_ox$ and the constraint gradients are normal to the tangent space. With $K_o=\bar U_o\Sigma_o\bar U_o\tr$, $U_o=P\tr\bar U_o$, and $z_o=\bar U_o\tr x_b$ this is a low-rank outage model, called the \emph{energy model} below. The change of the voltage magnitude at a PQ bus is the functional $\Delta V_i=\mathrm{Re}(\bar V_i\Delta V_i)/|V_i|$, and the linearized active and reactive branch flows give thermal criteria in the same model.

\begin{proposition}[Definiteness]\label{prop:congruent}
$H=T\tr J_0T$, where $J_0$ is the Jacobian of the lossless mismatches with respect to $(\theta_i)_{i\ne r}$ and $(\ln|V_i|)_{i\in\mathrm{PQ}}$ at $x_b$, which is symmetric. $T$ is the non-singular derivative of these coordinates with respect to $z$. Hence $H\succ0$ exactly when $J_0\succ0$, and $H$ is singular exactly where the lossless power-flow Jacobian is.
\end{proposition}
\begin{proof}
In the coordinates $(\theta,\ln|V|)$ the voltage constraints fix coordinates, and the gradient of $\mathcal W$ is the mismatch vector, so its Hessian $J_0$ is symmetric. At a critical point the Hessian of the Lagrangian on the tangent space transforms as a bilinear form under a change of chart, which gives $H=T\tr J_0T$, and $T$ is non-singular because $V\ne0$.
\end{proof}

The inertia of $H$ can only change where $H$ is singular, so $H\succ0$ holds along a loading path from a positive definite point up to the first saddle-node bifurcation. This is checked at the base of each system. The two traditional linear voltage models fall short because they each lack one of the three properties needed for this approach. In rectangular coordinates, one Newton step on the power-mismatch equations has at the base a non-orthogonal rescaling of the Hessian of $\mathcal W$ as its Jacobian, thus it is not symmetric even when lossless. In the IEEE 118-bus system, the normalized coupling exceeds one for a quarter of the branch contingency pairs. The decoupled $B''$ model is symmetric, but in this model the released reactive injection has a common part at both terminals outside the range of the branch term. In polar coordinates the branch Hessians pick up first-order terms. Hence, rectangular coordinates are used to make the network part of $\mathcal W$ exactly quadratic thus placing the released injection in the range of the branch term.

Removing the charging with the line makes $\Sigma_o$ indefinite.
Theorem~\ref{thm:lrbound}(ii) still applies whenever $c_{kj}<1$, which is known before the pair is solved. Pairs with $c_{kj}\ge1$ are treated as strong and evaluated. Algorithm~\ref{alg:search} runs per voltage criterion with base value zero, strong set $\{c_{kj}>\varepsilon\}$, weak scores $w_o=s\Delta V^o_i+\frac{\omega}{2}(t_i\|\gamma_o\|^2+\|s_o\|^2/t_i)$ with $\omega=\varepsilon/(1-\varepsilon)$ and any $t_i>0$, and strong caps from (\ref{eq:lrbound}), infinite when $c_{kj}\ge1$. Theorem~\ref{thm:correct} holds because it needs only valid bounds, and a pair's system is solved only when the pair is removed from the top of the priority queue.

\section{Adjustment Between Contingencies}\label{sec:redispatch}

After first contingency $k$, the adjustment $u$ (generation redispatch, phase-shifter angles, interruption of firm transfers, and load interruption) lies in a polyhedron $U_k$ defined by ramp and unit limits and power balance. With the PTDF matrix of the controls $H$, the post-$k$ flows are $f^k+H^{(k)}u$ with $H^{(k)}=H+d_{:,k}H_{k,:}$, and after the second outage
\begin{equation}
f^{k,u,j}_m=f^{k,u}_m+d^{(k)}_{mj}f^{k,u}_j,\qquad d^{(k)}_{mj}=\frac{d_{mj}+d_{mk}d_{kj}}{1-\kappa_{kj}} .
\end{equation}
Under a common adjustment, $k$ passes if one $u\in U_k$ satisfies every row $\sigma f^{k,u,j}_m\le R_m$. With a cost vector $\gamma$, in which load interruption carries a large penalty, the problem for $k$ is the linear program
\begin{equation}\label{eq:scrd}
\min_{u\in U_k}\ \gamma\tr u\quad\text{s.t.}\quad\sigma f^{k,u,j}_m\le R_m\quad\forall j\ne k,\ m,\ \sigma .
\end{equation}
The problems decouple over first contingencies, and are then solved by constraint generation: at a fixed $u$ the rows over $j$ are pair values in the post-$k$ network, whose cycle-space projector is the rank-one reduction $Q-q_kq_k\tr/Q_{kk}$. The most violated row is therefore found by Algorithm~\ref{alg:search}, and rows it does not reach are proven safe.

\begin{theorem}[Witness bound]\label{thm:helly}
Let $d=\dim\mathrm{aff}\,U_k$. (i) If (\ref{eq:scrd}) is infeasible, some set of at most $d+1$ rows is infeasible together with $U_k$. (ii) If (\ref{eq:scrd}) is feasible, some set of at most $d$ rows gives the same optimal value together with $U_k$.
\end{theorem}
\begin{proof}
Each row intersects $U_k$ in a convex subset of its affine hull, so (i) is Helly's theorem \cite{helly1923}. For (ii), at an optimal vertex the Karush--Kuhn--Tucker conditions express the projection of $-\gamma$ onto the direction space of $\mathrm{aff}\,U_k$ as a conic combination of active row normals and active constraints of $U_k$. Carath\'eodory's theorem reduces the combination to at most $d$ vectors, and the reduced multipliers prove optimality of the same vertex for the reduced program.
\end{proof}

The theorem bounds the second contingencies needed to show that $k$ fails, or to fix the adjustment cost; it does not reduce the work of showing that $k$ passes which is shown above. Passing with a common adjustment implies passing pair by pair, but not the converse. With $U_k=[-1,1]$, two second contingencies that require $u\ge\frac12$ and $u\le-\frac12$ each pass alone yet fail together.

\section{The AC Equations: Per-Pair Certificates and Open Gaps}\label{sec:ac}

The search excludes pairs it never evaluates, and every exclusion rests on a bound that is exact for a linear model. No such bound is presented for the AC equations; what remains for AC is a guarantee for each pair that is evaluated from the best-first search. Let $x$ be the rectangular bus voltages and $F_p$ the AC mismatch map after pair $p$, whose components are quadratic in $x$, so its second derivative $F_p''$ is constant.

For an evaluated pair the affine-covariant Newton--Kantorovich theorem \cite{deuflhard2004} gives such a guarantee. Let $J_p(x_c)$ be invertible, $s=-J_p(x_c)^{-1}F_p(x_c)$, $\eta=\|s\|$, and $\vartheta=\max_{\|\delta\|=1}\|J_p(x_c)^{-1}F_p''[\delta]\|$. If $h=\vartheta\eta\le\frac12$, then $F_p$ has a zero $x^*$ with
\begin{equation}\label{eq:nk}
\|x^*-x_c\|\le t^*=\frac{1-\sqrt{1-2h}}{\vartheta},\qquad \|x^*-(x_c+s)\|\le t^*-\eta .
\end{equation}
Because $F_p$ is quadratic, its Lipschitz condition holds globally, and $\vartheta\le\lambda_{\max}(\sum_aN_a\tr PN_a)^{1/2}$, where $P=J_p(x_c)^{-\mathsf T}J_p(x_c)^{-1}$ and $N_a$ are the constant Hessian slices. The radius bounds every bus voltage and, through the branch admittances, every branch apparent power. Section~\ref{sec:acresults} tests it on all pairs of the IEEE 118-bus system. The open gaps are:
\begin{enumerate}
\item \emph{No AC exclusion.} A uniform remainder of the form of Proposition~\ref{prop:residual} would require a bound on the AC Jacobian under a rank-two outage from single-outage quantities. The energy model of Section~\ref{sec:emodel} has this structure only for one Newton step from the base. The remainder of the AC equations beyond that step has no such bound, and the classical max-norm Kantorovich bound is vacuous on the IEEE 118-bus system (Section~\ref{sec:acresults}).
\item \emph{No link between the linear and the AC values.} Nothing bounds the difference between the DC flows or linearized voltages and the AC values of a pair that is not evaluated. Empirically the gap is small for flows, the energy model, and the power-mismatch linearization. The gap is large for the decoupled model (Section~\ref{sec:acresults}).
\item \emph{Large deviations and collapse.} Pairs far from the base state, including pairs with no AC solution, are the ones most understated by a model linearized at the base state. A guarantee for such pairs needs a bound on the solution set of $F_p$ over a region, not at a point.
\item \emph{Cost of the guarantee.} The guarantee (\ref{eq:nk}) needs one factorization of $J_p$ per pair, so it makes an AC candidate list exact, but does not reduce the number of AC solutions.
\item \emph{Discrete controls.} Generator reactive limits, tap changers, and switched shunts change $F_p$. Each switching state needs its own guarantee.
\end{enumerate}

\section{Case Studies}\label{sec:cases}

The best-first search was tested on five public systems distributed with MATPOWER \cite{zimmerman2011} and simulated through pandapower \cite{thurner2018}: the IEEE 118-bus system, the synthetic ACTIVSg200 and ACTIVSg2000 systems \cite{birchfield2017}, and the PEGASE 1354-bus and 2869-bus systems \cite{fliscounakis2013,josz2016}. $\E$ is formed from all non-bridge branches, all branches are monitored, and pairs with $1-\kappa_{kj}<10^{-9}$ are treated as islanding pairs.

Line ratings of the IEEE and PEGASE cases are not set for DC thermal studies (3\% to 4\% of PEGASE branches exceed their rating in the base case), so each branch was assigned the rating $R_m=1.25\max(F_m,F_{25})$, where $F_m$ is its largest flow over the N-1 outages and $F_{25}$ the 25th percentile of the $F_m$. The factor 1.25 was inherited from the related work performed by Kaplunovich \cite{kaplunovich2016}, while the floor was added to avoid unrealistically small ratings. To consider various loading scenarios, a loading factor $\lambda$ scales all injections with ratings fixed. Consequently, every case with $\lambda<1.25$ is N-1 secure. Section~\ref{sec:native} uses the native ratings of the ACTIVSg systems.

Brute force evaluated (\ref{eq:pair}) for every monitored branch and non-islanding pair. The unit of work is one evaluation of (\ref{eq:pair}) for one pair and one branch. Brute force evaluation determines both directions while the search may evaluate a flow once per direction, so the comparison is conservative for the search. Both methods are implemented in Python with NumPy and ran on a single core for accurate clock-time comparison. Brute force is vectorized, the search loop is interpreted, and times exclude the shared factorization. A default coupling value of $\varepsilon=0.02$ is used.

\subsection{Structure of the Pair Space}\label{sec:struct}

On every system the budget identity of Theorem~\ref{thm:budget} holds to within $3\times10^{-12}$. Strongly coupled pairs are a small fraction of all pairs, for example 6{,}904 of about one million in the PEGASE 1354 system at $\varepsilon=0.1$. The count (\ref{eq:count}) holds loosely, by a factor of 50 to 100.

Fig.~\ref{fig:residual} tests Proposition~\ref{prop:residual} on about 15{,}000 random triples $(k,j,m)$ per system. All samples lie below the bound and the upper envelope follows it closely with the largest ratio of interaction term to bound at 0.998 to 0.999 on the three smaller systems. On the PEGASE 2869-bus system the largest ratio was 0.65, because fewer than 0.3\% of the sampled pairs have $|c|>0.1$.

\begin{figure}[!t]
\centering
\includegraphics[width=\columnwidth]{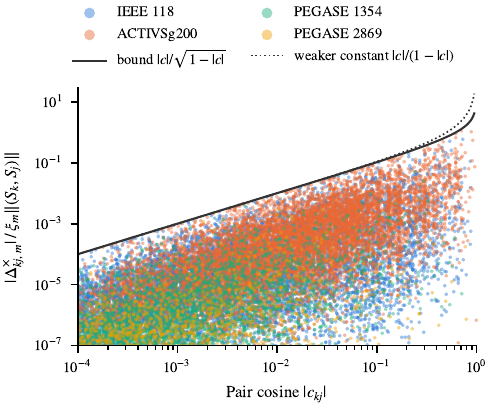}
\caption{Interaction term normalized by $\xi_m\|(S_k,S_j)\|$ versus the cycle-space cosine for random triples. The solid line is the bound of Proposition~\ref{prop:residual}; the dotted line is the weaker constant $|c|/(1-|c|)$ obtained from the entrywise form of (\ref{eq:pair}).}
\label{fig:residual}
\end{figure}

\subsection{Exactness and Size of the Minimum Set}

For the three smaller systems brute force was run for every evaluated loading factor and value of $\varepsilon$. The PEGASE 2869-bus system required $3.3\times10^{10}$ evaluations for a single loading factor and thus brute force was only performed for the reference loading. In every run the search returned the brute-force maximum for every violated criterion and found no violation elsewhere. At most six maximizing pairs were tied, and $|S^*|$ is the exact minimum hitting set of the tie sets as calculated by integer programming. Because the search returns one maximizer per criterion, its output is slightly larger (908 sequences against a minimum of 893 for PEGASE 2869).

Table~\ref{tab:sweep} shows that the minimum set is small. In the PEGASE 2869 system at $\lambda=0.8$, 893 sequences bound all 1{,}616 violated criteria out of approximately 7.2 million pairs. The set is smaller than the 1{,}736 violating pairs because a single pair is often the worst case for several branches. Table~\ref{tab:sweep} shows that the number of violating pairs grows faster than $|S^*|$ as the system is stressed. The ratio of violating pairs to $|S^*|$ rises from 1.6 to 3.5 in the PEGASE 1354-bus system.

\begin{table}[!t]
\centering
\caption{Effect of Loading ($\varepsilon=0.02$). For PEGASE 2869, Brute Force Was Run Only at $\lambda=0.8$}
\label{tab:sweep}
\footnotesize
\setlength{\tabcolsep}{3pt}
\begin{tabular}{lcrrrrr}
\toprule
System & $\lambda$ & Violated & Violating & $|S^*|$ & Search & Reduc- \\
 & & criteria & pairs & & evals. & tion \\
\midrule
IEEE 118 & 0.6 & 18 & 11 & 10 & 6{,}930 & 416 \\
 & 0.8 & 78 & 70 & 47 & 14{,}217 & 203 \\
 & 1.0 & 155 & 242 & 89 & 20{,}100 & 143 \\
\midrule
ACTIVSg200 & 0.6 & 50 & 45 & 20 & 13{,}659 & 264 \\
 & 0.8 & 120 & 132 & 48 & 22{,}813 & 158 \\
 & 1.0 & 184 & 344 & 70 & 28{,}640 & 126 \\
\midrule
PEGASE 1354 & 0.6 & 203 & 151 & 93 & 1{,}123{,}689 & 1{,}809 \\
 & 0.8 & 657 & 716 & 366 & 2{,}018{,}289 & 1{,}007 \\
 & 1.0 & 1{,}133 & 2{,}221 & 637 & 3{,}428{,}885 & 593 \\
\midrule
PEGASE 2869 & 0.8 & 1{,}616 & 1{,}736 & 893 & 14{,}922{,}670 & 2{,}221 \\
\bottomrule
\end{tabular}
\end{table}

\subsection{Search Effort}

Depending on the case the search needed approximately 130 to 3{,}800 times fewer pair evaluations than brute force (Table~\ref{tab:sweep}, Fig.~\ref{fig:reduction}). The reduction grows with system size, from about 200 for the 118-bus system to 2{,}200 for the PEGASE 2869 system at $\lambda=0.8$, and shrinks as the system is stressed. Before evaluation approximately 5\% to 30\% of criteria were discarded per the immunity test (\ref{eq:immunity}). The search also forms $|\E|$ scores and one cap per strong pair for each criterion, $1.1\times10^6$ bound computations on the IEEE 118-bus system and $7.0\times10^8$ on PEGASE 2869. In wall-clock time it matches the vectorized brute force on the two small systems and is about 5, 23, and 7 to 9 times faster on PEGASE 1354, PEGASE 2869, and ACTIVSg2000.

\begin{figure}[!t]
\centering
\includegraphics[width=\columnwidth]{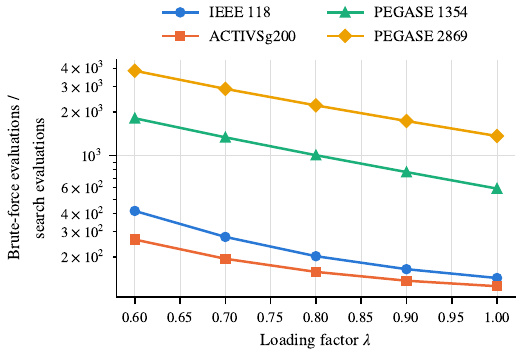}
\caption{Ratio of brute-force to search evaluations versus loading factor ($\varepsilon=0.02$).}
\label{fig:reduction}
\end{figure}

Fig.~\ref{fig:eps} shows the effect of $\varepsilon$: every value tested returned identical worst-case values, and the effort is lowest near $\varepsilon=0.01$--$0.02$. The evaluations per searched criterion grow from about 40 on the 118-bus system to about 2{,}000 on the PEGASE 2869 system. The bound of Proposition~\ref{thm:work} held for every criterion at $\varepsilon=0.1$ but is loose by orders of magnitude, and the computable count after it by a median factor of 6 to 80.

\begin{figure}[!t]
\centering
\includegraphics[width=\columnwidth]{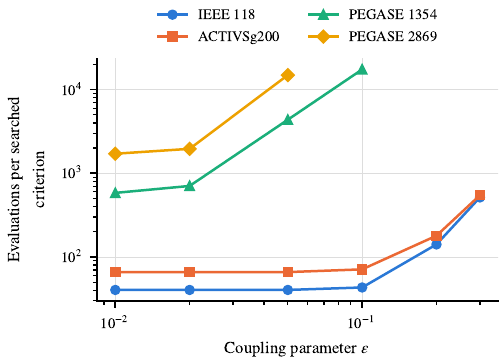}
\caption{Evaluations per searched criterion versus $\varepsilon$ at $\lambda=0.8$. The worst-case values are identical for every $\varepsilon$.}
\label{fig:eps}
\end{figure}

\subsection{Comparison With the Filter of \texorpdfstring{\cite{kaplunovich2016}}{Kaplunovich and Turitsyn}}\label{sec:kt}

The filter as described in Kaplunovich's work was implemented for comparison \cite{kaplunovich2016}. Differing from the best-first search that finds the worst pair of each violation criterion, the filter aims to list all violating pairs by removing monitored branches and second outages shown safe by LODF-derived bounds until neither set changes. The remaining set is then evaluated exactly. The filter found every violating pair for each system and loading condition.

Neither method needs fewer evaluations in every case: the filter is better at light loading and the search under stress. At $\lambda=0.6$ the filter performs between a third and half as many evaluations as the search on the two small systems, and roughly one-eighth as many on PEGASE 1354. It gains by removing entire first contingencies that bring no branches near their rating. However, the number of candidate sets from the filter grows quickly as system loading increases: at $\lambda=1.0$ it needs 24, 25, and 82 times more evaluations than the search on the IEEE 118, ACTIVSg200, and PEGASE 1354 systems, and 83 times more on PEGASE 2869 at $\lambda=0.8$. On the PEGASE systems the speed was more comparable to brute force (54 to 72~s on PEGASE 1354, 1{,}355~s on PEGASE 2869), while the search took 6 to 14~s and 64~s respectively. Counting bound computations, the search needed 2.5 to 84 times fewer operations than the filter.

\subsection{Native Ratings, Several Study Cases, and Ranking Rules}\label{sec:native}

The synthetic ratings guarantee an N-1 margin; as a check, the native RATE\_A of the ACTIVSg systems was used instead. On the ACTIVSg200 system at $\lambda=1.0$, 1.1, and 1.2, the minimum sets had 6, 12, and 15 sequences and the search needed 1{,}187 to 2{,}151 times fewer evaluations than brute force. On the ACTIVSg2000 system at $\lambda=1.0$, nine branches are already overloaded by some single outage, so Proposition~\ref{thm:work} does not apply to them. There, 673 sequences bound 961 violated criteria. The search needed 5.6 million evaluations, approximately 2{,}200 times fewer than brute force, and ran about 9 times faster, while the filter needed $8.5\times10^8$. At $\lambda=1.1$, where every non-islanding pair violates some criterion, 885 sequences bound all criteria, and the search remained exact.

For four study cases per system (loading factors 0.6, 0.8, and 1.0, and an alternative dispatch at $\lambda=0.8$), the envelope set (or minimum hitting set) over all cases had 109, 78, and 758 sequences for the IEEE 118, ACTIVSg200, and PEGASE 1354 systems against sums of the per-case sets of 210, 192, and 1{,}524. The search remained exact. To quantify the ranking-rule effect, each maximizing pair at $\lambda=0.8$ was compared against the ranking of single-outage impacts on the same branch: a rule that pairs the ten most severe single outages of each branch would not identify the worst sequence for 6 of 78 violated criteria in the IEEE 118-bus system, 40 of 120 in the ACTIVSg200 system, and 51 of 657 in the PEGASE 1354 system. In the PEGASE 1354 system, one worst pair contains the element traditionally ranked 1{,}428th of 1{,}429 by its single impact.

\subsection{Adjustment Between Contingencies}

Problem (\ref{eq:scrd}) was solved for every first contingency of the IEEE 118-bus system with the 54 generators as controls, each limited to 10\% of its capacity, minimizing the total redispatch in MW. A second run implemented load interruption at the 99 load buses at 1{,}000 per MW. The witness of a secured first contingency is the set of second contingencies containing nonzero multipliers. For a first contingency that fails, the witness is an irreducible, infeasible subsystem found by a deletion filter \cite{chinneck2008}. At $\lambda=0.8$, 0.9, and 1.0, 100, 69, and 44 first contingencies pass without adjustment. Redispatch secured 54, 70, and 65 first contingencies respectively, while 23, 38, and 68 fail. No first contingency needed more than four second contingencies for a pair to be determined with redispatch only, against a Helly bound of 54. Five at $\lambda=1.0$ fail with a common adjustment although each second contingency in the witness passes alone. Load interruption secured all but 4, 7, and 17 of the failures with at most 144~MW interrupted. Remaining failures are ramp-limited as every first contingency was secured with unrestricted ramping. Witness sets increased to at most eight second contingencies, against a Helly bound of 153.

\subsection{Voltage Criteria}\label{sec:vresults}

The energy model was constructed at the solved AC case of each system at $\lambda=1.0$, with line charging included in line removal and $V^{\min}=0.95$ and $V^{\max}=1.05$~p.u.\ at the PQ buses. Brute force evaluated Theorem~\ref{thm:lrpair}(ii) for every PQ bus and non-islanding pair. The pair formula agreed with direct solution of (\ref{eq:lrmodel}) to $10^{-15}$. $H$ was positive definite on each system, with smallest eigenvalue between 0.011 (ACTIVSg2000) and 0.198 (IEEE 118). The inertia test Theorem~\ref{thm:lrpair}(i) flagged no pair on three systems, one single outage with its 1{,}427 pairs and one pair of parallel lines on PEGASE 1354, and 12 pairs on ACTIVSg2000. For every flagged outage, Newton's method from the base fails in AC. Flagged pairs and base case violations were excluded from both brute force and search.

Theorem~\ref{thm:lrbound} was checked on every remaining pair and every PQ bus, roughly $2.4\times10^{10}$ pair--bus combinations. The bound held up to rounding ($10^{-13}$~p.u.), and the largest ratio of interaction term to bound was 0.96 to 0.9999996. For most pairs the coupling is weak: the median of $c_{kj}$ is 0.0004 to 0.014 and its 99th percentile 0.03 to 0.56, and $c_{kj}\ge1$ occurred only for flagged pairs.

\begin{table*}[!t]
\centering
\caption{Voltage Criteria in the Energy Model; the Sets Bound the Model, Not AC. $\lambda=1.0$, $V^{\min}=0.95$, $V^{\max}=1.05$~p.u., $\varepsilon=0.02$. ``Solved'' Is the Share of Pairs Whose System of Theorem~\ref{thm:lrpair} the Search Solved. Search Time Includes the Couplings of All Pairs. Times in Seconds}
\label{tab:voltage}
\setlength{\tabcolsep}{3.4pt}
\begin{tabular}{lrrrrrrrrrrrr}
\toprule
System & PQ & Violated & Violating & $|S^*_V|$ & Strong & Solved & Search & Brute-force & Reduction & Search & Brute-force & Mis- \\
 & buses & criteria & pairs & & pairs & & evaluations & evaluations & factor & time & time & matches \\
\midrule
IEEE 118 & 64 & 34 & 3{,}205 & 25 & 25.1\% & 28\% & $5.9\times10^{3}$ & $9.9\times10^{5}$ & 168 & 0.4 & 0.2 & 0 \\
ACTIVSg200 & 162 & 119 & 1{,}130 & 25 & 42.6\% & 5\% & $9.2\times10^{3}$ & $2.4\times10^{6}$ & 259 & 0.3 & 0.2 & 0 \\
PEGASE 1354 & 1{,}094 & 418 & 218{,}458 & 162 & 3.8\% & 18\% & $3.4\times10^{5}$ & $1.1\times10^{9}$ & 3{,}273 & 22.8 & 61.6 & 0 \\
ACTIVSg2000 & 1{,}611 & 1{,}205 & 383{,}993 & 532 & 6.4\% & 3\% & $2.6\times10^{5}$ & $6.1\times10^{9}$ & 23{,}680 & 68.3 & 427.2 & 0 \\
PEGASE 2869 & 2{,}359 & 960 & 1{,}560{,}134 & 344 & 1.7\% & 12\% & $1.2\times10^{6}$ & $1.7\times10^{10}$ & 14{,}690 & 117.9 & 1181.8 & 0 \\
\bottomrule
\end{tabular}
\end{table*}

Results are reported in Table~\ref{tab:voltage}. The search reproduced the brute-force worst value and full tie set of every violated voltage criterion, and reported no violation elsewhere, on every system at $\lambda=1.0$, at $\lambda=0.8$ on the two smaller systems, and for $\varepsilon = 0.02$, 0.05, and 0.1 on the IEEE 118-bus system. The search needed 170 to 23{,}700 times fewer pair--bus evaluations, and it solved the system of Theorem~\ref{thm:lrpair} for only 3\% to 28\% of the pairs because the weak scores and coupling require only entries of $H^{-1}$ and single-outage quantities. In contrast, a linearization without an a priori bound must solve every pair first. The identified minimum sets are small, 25 to 532 sequences for 34 to 1{,}205 violated criteria. With at most two maximizing pairs, ties were rare. The search was 2.7 to 10 times faster than the vectorized brute force on the three large systems. The strong set represents a quarter to two-fifths of the total pairs on the two small systems, but only a few percent of the pairs on the large systems.

\subsection{The AC Equations}\label{sec:acresults}

Every non-islanding pair of the IEEE 118-bus system was solved with the full AC equations at $\lambda=0.8$ and $1.0$ (loads and generation scaled by $\lambda$, generators at fixed voltage set points without reactive limits, constant-power loads, fixed taps), with MVA ratings from the rule of Section~\ref{sec:cases} applied to the AC N-1 apparent power. The guarantee (\ref{eq:nk}) was applied at the linear prediction one Newton step from the base, with additional Newton steps where $h>\frac12$, and the intervals were adjusted until each criterion had a single candidate worst pair.

Three findings relate directly to the gaps of Section~\ref{sec:ac}. Primarily, the per-pair guarantee is applicable and tight: it held for every pair but one (at the linear prediction for 7{,}118 and 5{,}161 pairs at the two loadings, after one to three Newton steps for the rest), the guaranteed radius exceeded the actual error by a median of 1--2\%, and never by more than a factor of two. The guarantees alone reproduced the AC minimum sets: 41 and 77 thermal, and 12 and 22 low-voltage sequences. The one pair lacking a guarantee, the loss of lines from bus 45 to buses 46 and 49 at $\lambda=1.0$, indicates voltage collapse as Newton's method failed. The classical max-norm remainder is vacuous: its Lipschitz constant is $\ell=3{,}095$, and on 200 random pairs the Kantorovich quantity had a median of 52, with only one pair below $\frac12$.

Second, for branch loadings the DC worst pair is typically the AC worst pair. The AC worst pair matched rank 1 of the DC ordering for 60 of 67 and 109 of 136 violated branches, and within the top five for all but one branch. The exception, sitting at DC rank 168, is an overload driven by voltage depression as opposed to active power redistribution. Where the DC pair differs, it can understate the AC worst loading by up to 21\%. The energy model also carries thermal criteria (Theorem~\ref{thm:lrbound} held, ratio up to 0.98), but it ranks the AC worst pair first for 59 of 67 and 98 of 136 violated branches, against the DC model's 60 and 109.

Third, none of the three linear voltage models bounds the AC voltages, and the energy model is the closest (Table~\ref{tab:vac}). Evaluation was performed on the AC solutions of all pairs, on the criteria of Section~\ref{sec:vresults}. For the two loading conditions, the energy model misses 6 of the 53 AC-violated criteria, against 10 and 26 for the power-mismatch linearization and the decoupled model respectively. The energy model flags 13 criteria that AC does not violate, performing conservatively. The energy model's correlation with the AC voltage changes of 0.97 to 0.98 is close to that of the power-mismatch linearization. Both outperform the decoupled model in this regard as it omits the reactive losses of the redistributed active power. Half of the energy model's misses involve the loss of the transformer from bus 5 to bus 8, which produces a strongly nonlinear response that also causes the power-mismatch linearization to fail. Remaining misses are understated overvoltages.

The largest difference is near system collapse. The loss of lines from bus 45 to buses 46 and 49 depresses bus 45 by 0.22~p.u.\ at $\lambda=0.8$ and has no AC solution at $\lambda=1.0$. The energy model predicts drops of 0.36 and 0.61~p.u.\ and gives this pair the largest coupling of the system, $c_{kj}=0.85$ and 0.87, while the power-mismatch linearization predicts 0.079~p.u.\ and the decoupled model 0.006~p.u. Thus, the energy model overstates the pair most understated by the traditional methods. Every outage the inertia test flags on PEGASE 1354 and ACTIVSg2000 is one for which Newton's method fails.

\begin{table}[!t]
\centering
\caption{Linear Voltage Models Against the AC Solutions of All Pairs, IEEE 118-Bus System ($V^{\min}=0.95$, $V^{\max}=1.05$~p.u.)}
\label{tab:vac}
\footnotesize
\setlength{\tabcolsep}{2.5pt}
\begin{tabular}{p{4.6cm}cc}
\toprule
 & $\lambda=0.8$ & $\lambda=1.0$ \\
\midrule
AC-violated (low/high); AC $|S^*_V|$ & 23 (21/2); 13 & 30 (29/1); 23 \\
\midrule
\multicolumn{3}{l}{\emph{Energy model; violated / $|S^*_V|$: 26/19 and 34/26}} \\
AC violations missed / false flags & 3 / 6 & 3 / 7 \\
AC worst pair at rank 1 / in top 5 & 20 / 21 & 23 / 26 \\
Correlation of $\Delta V$ with AC & 0.980 & 0.967 \\
\midrule
\multicolumn{3}{l}{\emph{Power-mismatch linearization; 22/14 at both}} \\
AC violations missed / false flags & 1 / 0 & 9 / 1 \\
AC worst pair at rank 1 / in top 5 & 20 / 22 & 24 / 26 \\
Correlation of $\Delta V$ with AC & 0.985 & 0.975 \\
\midrule
\multicolumn{3}{l}{\emph{Decoupled model; 13/8 and 17/11}} \\
AC violations missed / false flags & 11 / 1 & 15 / 2 \\
AC worst pair at rank 1 / in top 5 & 18 / 21 & 19 / 23 \\
Correlation of $\Delta V$ with AC & 0.836 & 0.769 \\
\bottomrule
\end{tabular}
\end{table}

\section{Discussion and Limitations}\label{sec:discussion}

The bounding sets presented in this paper are properties of linear models for thermal and voltage criteria. Against nonlinear (AC) models they are strong ordered candidate lists, not bounds. The guarantee (\ref{eq:nk}) makes such a list exact. For the linear model an a priori theory is determined from the form: the energy model and the traditional power-mismatch linearization provide similar AC prediction, but only the energy model has a symmetric system matrix and semi-definite branch terms whose ranges contain the released injection. That form yields the coupling cosine, the separable bound, and the collapse test.

The findings suggest a route to the AC equations that maintains the guarantees: the AC value of a contingency pair is its energy-model value plus a remainder. Because the network part of $\mathcal W$ is exactly quadratic, the remainder comes only from the bus-local load and voltage-control terms. On the IEEE 118-bus system the largest voltage error of a pair grows with the a priori quantity $\|s\|/\mu$, with $\mu$ the smallest eigenvalue magnitude of $C$. Inspecting the deciles of this quantity, the median error rises from $3\times10^{-4}$ to $5\times10^{-3}$ and $8\times10^{-3}$~p.u.\ at the two loadings (rank correlation 0.45 and 0.49, against 0.15 and 0.16 for $c_{kj}$ alone), and the pair from bus 45 ranks in the top ten of 15{,}502. The correlation is moderate, so a uniform bound in these quantities alone would be loose. A remainder bound that also uses the exposures $\gamma_o$ would let Algorithm~\ref{alg:search} exclude AC pairs, with the guarantee (\ref{eq:nk}) certifying the rest. A ranking fitted to AC brute-force results could order that verification, but it cannot certify the pairs it skips. Thus solely fitting AC brute-force results likely cannot supply the documented rationale required by planning standards.

Mirroring the DC model, the energy model drops conductances. A lossy version of this model is not symmetric and hence a bound for it is open. There is no a priori count of strongly coupled voltage pairs, and for flows the count is loose. In the energy model, $c_{kj}<1$ is checked rather than guaranteed for outages that remove line charging. This work evaluates only branch outages. Shunt-device outages fit the energy model form as bus local terms, but sequences including them were not studied here, nor were DC poles. The AC study ignores reactive limits, tap changes, and switched shunts, and the adjustment study was only performed on a single system. Future work should extend the adjustment study to more systems and control models.

\section{Conclusion}\label{sec:conclusion}

In linear network models, the smallest set of N-1-1 branch-outage pairs, bounding for all thermal and voltage criteria, can be computed exactly without evaluating every pair. The interaction between two outages is bounded by a coupling cosine in a DC model for thermal criteria evaluation and a lossless energy model for voltage criteria. In each model, outages remove semi-definite branch terms from a symmetric positive definite system matrix. A best-first search built on these bounds is exact and instance-optimal. The impacts of post-first contingency system adjustment are explored and tested, with Helly's theorem limiting the second contingencies that determine each first contingency under the allowed system adjustment. On public systems of up to 2{,}869 buses, the search reproduced the brute-force worst cases with up to 3{,}800 times fewer pair evaluations for thermal criteria and up to 23{,}700 times fewer pair--bus evaluations for voltage criteria. Comparison against AC shows that the search-determined sets are strong candidate lists but cannot be guaranteed. The suggested route to excluding AC pairs without evaluation is the creation of a remainder bound in the pair quantities of the energy model.

\bibliographystyle{IEEEtran}
\bibliography{references}
\end{document}